\documentclass[conference]{IEEEtran}
\IEEEoverridecommandlockouts
\usepackage{cite}
\usepackage{amsmath,amssymb,amsfonts}
\usepackage{graphicx}
\usepackage{textcomp}
\usepackage{xcolor}
\usepackage{amsthm}
\usepackage{overpic}
\usepackage{bm}
\usepackage{overpic}
\usepackage{float}
\usepackage{amsthm}
\usepackage{algorithm,algorithmic}

\usepackage[caption=false,font=footnotesize]{subfig}

\theoremstyle{plain}
\newtheorem{theorem}{Theorem}

\makeatletter
\newcommand\fs@ruled@notop{\def\@fs@cfont{\bfseries}\let\@fs@capt\floatc@ruled
  \def\@fs@pre{}%
  \def\@fs@post{\kern2pt\hrule\relax}%
  \def\@fs@mid{\kern2pt\hrule\kern2pt}%
  \let\@fs@iftopcapt\iftrue}
\renewcommand\fst@algorithm{\fs@ruled@notop}
\makeatother

\begin{document}

\title{Energy-Efficient Dual-Band Communication:\\ How to Allocate Traffic to Sub-THz Carriers?}

\author{\IEEEauthorblockN{Anders Enqvist, Emil Bj{\"o}rnson, Cicek Cavdar}
\IEEEauthorblockA{\textit{Department of Communication Systems, KTH Royal Institute of Technology, SE-100 44 Stockholm, Sweden} \\
Email: enqv@kth.se, emilbjo@kth.se, cavdar@kth.se} \thanks{This work was supported by the Swedish Innovation Agency (Vinnova) through the SweWIN center (2023-00572).}
}

\IEEEpubid{\makebox[\columnwidth]{979-8-3315-6128-4/26/\$31.00~\copyright{}2026 IEEE \hfill} \hspace{\columnsep}\makebox[\columnwidth]{ }}

\maketitle

\begin{abstract}
As 6G wireless networks transition toward sub-Terahertz (sub-THz) frequencies to satisfy extreme capacity demands, managing the trade-off between massive bandwidth and power consumption becomes a critical design challenge. In this paper, we investigate the fundamental energy efficiency (EE) limits of a dual-band base station site combining a coverage-oriented sub-6 GHz carrier with a capacity-oriented sub-THz carrier. By jointly optimizing hardware parameters and advanced sleep modes via activity factors, we identify four distinct operational regions that govern the EE-optimal behavior across the complete range of data rates. We derive closed-form analytical thresholds that dictate precisely when the sub-THz band should awaken from sleep and how to allocate traffic between the bands in that case. Our results demonstrate that utilizing the sub-THz band is EE-optimal when the power cost of the bandwidth-limited sub-6 GHz band surpasses the static power penalty of activating the sub-THz circuitry. Ultimately, this framework provides mathematically rigorous guidelines for power consumption minimization and sleep-mode management in future green networks.
\end{abstract}


\section{Introduction}

The growth in traffic within wireless communication systems, as dictated by Cooper's law, has substantially raised power consumption (PC) in recent years. Mobile operators consumed approximately $290$~TWh of electricity in 2023, with the Radio Access Network (RAN) alone accounting for up to $87\%$ of this ~\cite{kolta2024measuring}. Anticipating continued growth in traffic demands~\cite{ericssonmobilityreport2023}, it is imperative to gradually enhance the energy efficiency (EE), defined as the data rate divided by the related PC \cite{bjornson2017massive}. By using EE as a primary optimization metric, we can design innovative system configurations that promise more sustainable and efficient wireless networks in the future. Much of the technology development has focused on increasing data rates through the introduction of massive MIMO (multiple-input multiple-output) \cite{bjornson2017massive} and increasing bandwidth in mmWave and sub-Terahertz (sub-THz) bands \cite{Rappaport2019a}.
Additionally, emerging technologies such as Reconfigurable Intelligent Surfaces (RIS) \cite{RIS-EE} promise reductions in PC in future networks. A recent survey on how these and other techniques can lead to PC reductions can be found in \cite{lopez2022survey}.

\subsection{Prior Work and Motivations}
Time domain sleep modes have become a staple mechanism for saving power in conventional sub-6 GHz networks \cite{Tombaz2014a,lopez2022survey}. Recent literature has explored single-band EE ``rush-to-sleep'' strategies \cite{rottenberg2024information} and numerical optimization of time, space, and power resources \cite{peschiera2025optimizing}. However, single-band architectures are fundamentally limited by available bandwidth, prompting a transition toward sub-THz frequencies when the very high rates are required \cite{Rappaport2019a}. While dual-band carrier aggregation (CA) is extensively studied as a method to maximize peak throughput \cite{tataria20216g}, and foundational PC models for aggregated carriers have been proposed \cite{lopez2021energy}, its application for PC minimization remains an open problem. Current literature lacks a rigorous analytical framework for coordinating sleep modes across non-contiguous dual-band CA architectures. Specifically, the data rate threshold dictating when a BS should awaken a secondary carrier to relieve a saturated primary band has not yet been established.

\subsection{Contributions}

In this paper, we aim to answer the research questions:

\begin{itemize}
    \item At what rates is it energy-efficient to activate and utilize a secondary sub-THz carrier in a dual-band system?
    \item When both bands are active, how should the traffic be allocated and the BS's hardware be configured in the respective bands to minimize the PC?
\end{itemize}

We answer these questions rigorously through a mathematical analysis that establishes four different operating regions, and then demonstrates the results numerically.

\section{System Model} \label{sec:system_model}
As illustrated in Fig.~\ref{fig:system_model}, we consider how to transfer downlink data between a dual-band base station (BS) and a multi-antenna user equipment (UE) to maximize EE. The BS operates over two widely spaced carriers (with or without shared hardware), denoted by $i \in \{1, 2\}$, representing a conventional sub-6 GHz band and a high-bandwidth sub-THz band. For each band $i$, the BS utilizes $M_i$ transmit antennas, the UE has $N_i$ receive antennas, and the carrier bandwidth is $B_i$. 
We consider a far-field Line-of-Sight (LoS) propagation scenario where the sub-THz band can be used efficiently if a sufficient number of antennas are used at both ends of the link to overcome the severe pathloss per antenna.

Under these conditions, the channel matrix between the BS and the UE has rank one in both bands.

Assuming perfect channel state information (CSI), which is relatively easy to achieve in LoS scenarios, the optimal transmission strategy is to use maximum ratio transmission (MRT) at the BS and maximum ratio combining (MRC) at the UE. This joint beamforming yields a total array gain of $M_i N_i$, and can be achieved using either analog, hybrid, or digital beamforming. Therefore, the data rate during active transmission in band $i$ is
\begin{equation}
    R_i = B_i \log_2 \left( 1 + \frac{P_i M_i N_i \beta_i}{B_i N_0} \right) \quad \textrm{bit/s},
\end{equation}
where $P_i$ is the transmit power, $\beta_i$ is the pathloss coefficient, and $N_0$ is the power spectral density of the receiver noise.

\begin{figure}[t!]
	\centering 
	\begin{overpic}[width=.99\columnwidth,tics=10]{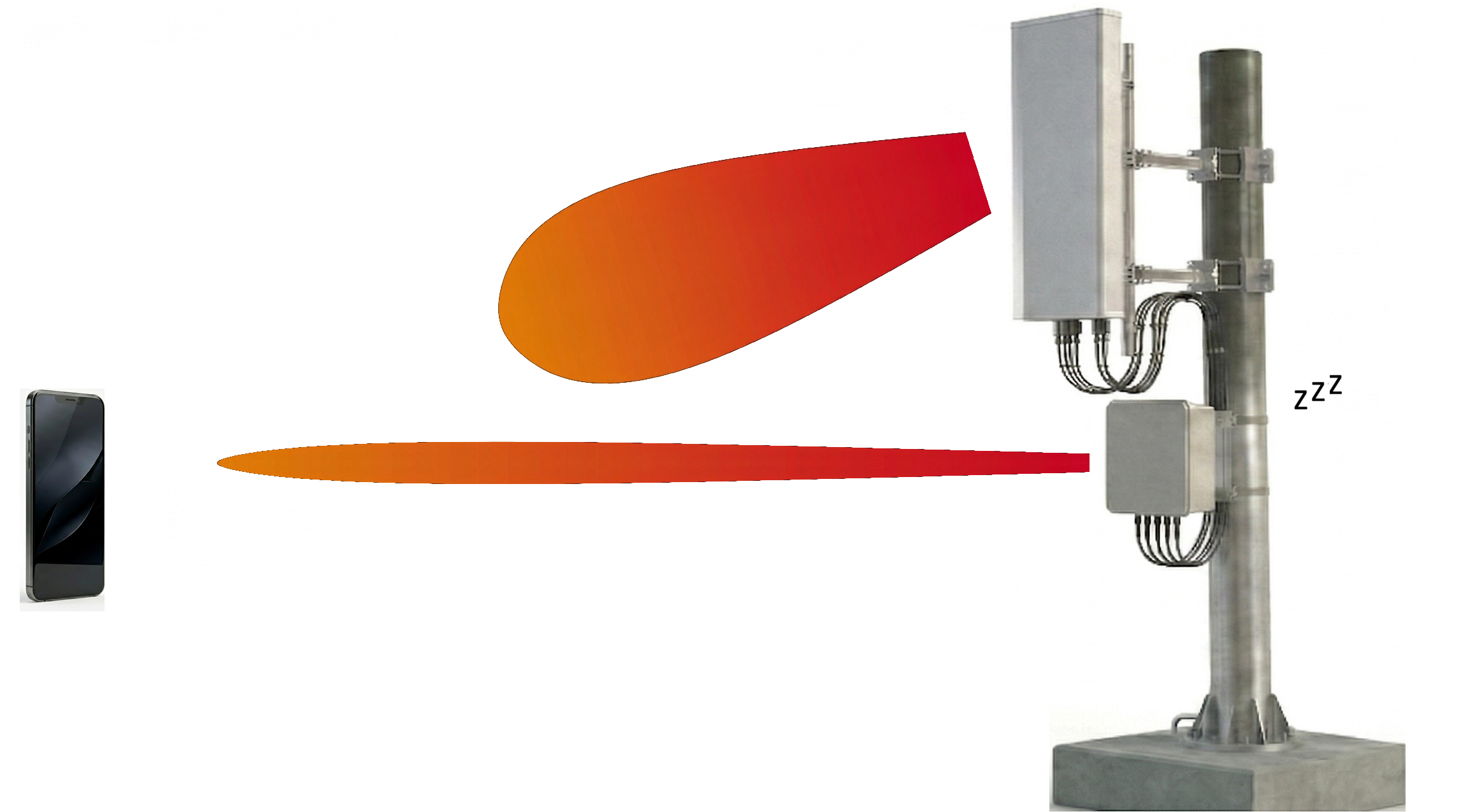}
		\put(2,10.5){\footnotesize UE}
		\put(89,10.5){\footnotesize  BS}
		\put(45,46){\footnotesize $\beta_1$}
        \put(45,18){\footnotesize $\beta_2$}
\end{overpic}  \vspace{-2mm}
	\caption{System model of the dual-band base station downlink. To maximize energy efficiency, the multi-antenna BS coordinates transmissions to a UE across two frequency bands, utilizing sleep modes and waking the secondary band only as needed. The beams represent beamforming gain.}  \vspace{-3mm}
\label{fig:system_model}
\end{figure}
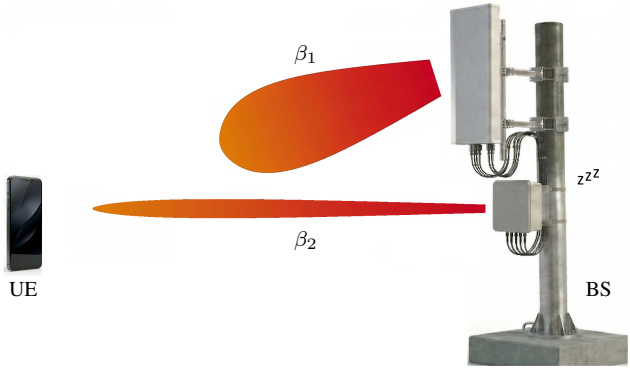

Suppose the UE requires $R$ bit/s. Under bursty or stochastic traffic, it represents the average throughput over a specific scheduling interval. This rate can be divided across the two bands in many different ways so that $R=R_1+R_2$. We seek to find the solution that maximizes EE, which corresponds to minimizing the BS's total PC. To this end, we will treat $P_i$ and $M_i$ as optimization variables for $i \in \{1, 2\}$, allowing us to control the utilized power and active antennas subject to constraints to be defined later.\footnote{The number of receive antennas $N_i$ is constant since we only measure the BS power consumption, as it dominates the total power consumption.}
Note that $P_i M_i=0$ corresponds to a deactivation of band $i$ (i.e., sleep mode).

\subsection{Power Consumption}

To optimize the system, we need a detailed model of the BS's PC that depends on our optimization variables and essential hardware coefficients.
When actively transmitting using band $i$, the associated PC is modeled as \cite{enqvist2024fundamentals}\footnote{The expression could also contain the term $\eta_i R_i$, where $\eta_i$ is a coefficient that determines the PC related to signal encoding. However, as established in our previous work \cite{enqvist2024fundamentals}, this term does not alter the EE-maximizing hardware configuration and has little impact on the total PC. To keep the notation simple, we have omitted this term from the optimization objectives in this paper.}
\begin{equation} \label{eq:Pai}
P_{a,i} = \frac{P_i}{\kappa_i} + \mu_i + (D_{0,i} + \nu_i B_i)M_i,
\end{equation}
where $\kappa_i \in (0, 1]$ is the PA efficiency and $\mu_i$ represents the fixed load-independent power required for circuitry and synchronization.  Moreover, $D_{0,i}$ is the power consumed by each transceiver chain, and $\nu B_i$ represents the signal-processing-related PC that scales with the sampling rate.

Depending on the rate requirement, the BS might switch between active and sleep modes to minimize the PC. 
In the sleep mode, the BS only consumes the passive circuit power required to keep the band powered, meaning the PC during sleep mode in band $i$ is 
\begin{equation} \label{eq:P_S}
    P_{s,i} = \mu_i.
\end{equation}

We let $\alpha_i \in [0, 1]$ denote the activity factor (or duty cycle) defined as the proportion of time that band $i$ is actively used for transmission, and treat it as an optimization variable \cite{peesapati2021analytical}. The BS stays in sleep mode for the remaining $(1-\alpha_i)$ fraction of the time. Consequently, the average PC for band $i$ is
\begin{equation} \label{eq:Pavg}
P_{\mathrm{avg},i} = \alpha_i P_{a,i} + (1 - \alpha_i)P_{s,i}.
\end{equation}

Because the two bands operate at widely separated frequencies and use distinct hardware components, their PC models are decoupled.
Hence, we can treat the bands independently at the physical layer while jointly managing their activity factors and active hardware configurations ($\alpha_i,P_i,M_i$), to determine the optimal EE configurations for the entire dual-band system.

\section{Dual-Band Power Consumption Minimization}

This paper considers the novel problem of minimizing the total PC when having two non-contiguous frequency bands (i.e., a sub-6 GHz carrier and a sub-THz carrier).
Before treating this dual-band problem, we will recap known results from previous works on single-band EE maximization.

\subsection{Preliminaries on Single-Band EE Maximization}

Our previous works \cite{enqvist2024fundamentals,Enqvist_Fundamentals_TBS} established that, in practice, the PC in a single band is minimized by using all the available bandwidth and a configuration $(P_i^*, M_i^*)$ that satisfies the optimal power-per-antenna ratio
\begin{equation} \label{eq:optimal_ratio}
M_i^* = \frac{P_i^*}{\kappa_i(D_{0,i} + \nu B_i)},
\end{equation}
which makes the first and last terms in \eqref{eq:Pai} equally large. The transmit power $P^*_i$ is then given by 
\begin{equation} \label{eq:Pstar}
    P^*_i=\sqrt{\frac{u^* - 1}{c_i}},
\end{equation}
where $u^* = \frac{-2}{W_0( -2 e^{-2} )}\approx 4.9216$ is a numerical constant, $W_0(\cdot)$ denotes the Lambert function, and

\begin{equation}
    c_i = \frac{N_i\beta_i}{\kappa_i (D_{0,i} + \nu_i B_{i}) B_{i} N_0}.
\end{equation}

When operating continuously at this optimal configuration, we achieve the following data rate and EE:
\begin{align} \label{eq:optimal_rate}
    R_{i}^* &=B_i \log_2 \left( 1 + \frac{P_i^* M_i^* N_i \beta_i}{B_i N_0} \right)=B_i \log_2(u^*), \\
    \mathrm{EE}_{i,\mathrm{opt}} & = \frac{R_{i}^*}{2\frac{P_i^*}{\kappa_i} + \mu_i}.
\end{align}

If we only use this band and the requested data rate $R$ is lower than $R_{i}^*$, then the PC is minimized by having an activity factor $\alpha_i = R/R_{i}^*$ and using the sleep mode during the rest of the time. This allows the average PC to be minimized by scaling it linearly with the required rate up to the value $R_i^*$. Building upon these single-band foundations, we now introduce our novel framework for dual-band EE.

\subsection{Problem Formulation: When to Turn on the Second Band?}

The two bands will generally have different optimal EE values: $\mathrm{EE}_{1,\mathrm{opt}}\neq \mathrm{EE}_{2,\mathrm{opt}}$.
In practical scenarios, due to the more favorable propagation conditions and higher hardware efficiency in the sub-6 GHz band, we assume that the coefficients $\kappa_i,D_{0,i},\beta_i,\nu_i$ are such that $\mathrm{EE}_{1,\mathrm{opt}}> \mathrm{EE}_{2,\mathrm{opt}}$.
This implies that when the requested data rate $R$ is relatively low, it is most efficient to only use the sub-6 GHz carrier.
In particular, if the data rate $R \leq R_{1}^*$, we say that the BS operates in \textbf{Region 1} and only uses Band~1 with an activity factor of $\alpha_1 = \min (1,R/R_{1}^*)$. For larger rates, Band~1 will ultimately be limited by its narrow bandwidth. 
This raises the fundamental question:

\emph{When minimizing the PC in a dual-band system, at what data rate threshold is it optimal to turn on the second (sub-THz) band and use both bands simultaneously?}

The answer can be found by solving the following total power minimization problem:
\begin{equation}
\begin{aligned}
\underset{\alpha_1,\alpha_2,P_1,P_2,M_1,M_2}{\text{minimize}} \,\,\, & \sum_{i=1}^{2} P_{\mathrm{avg},i}  \\
\text{subject to} \quad \quad & \sum_{i=1}^{2} \alpha_i R_i= R \\
& 0 \leq P_i , \,0 \leq M_i ,\,  0\le \alpha_i \le 1,
\end{aligned}
\end{equation}
where the first constraint ensures that the required rate is achieved, the second constraint ensures non-negative power and antenna allocations, and the third constraint is for the activity factors.

This problem can be simplified using the preliminary results from the single-band case.
By substituting the optimal power-per-antenna ratio from \eqref{eq:optimal_ratio} into $P_{a,i}$, we can eliminate the optimization variable $M_i$ and the average PC for band $i$ simplifies to $P_{\mathrm{avg},i} = \alpha_i \frac{2P_i}{\kappa_i} + \mu_i$.
This reduction yields the following simplified power allocation problem:

\begin{equation} \label{eq:quad_power_allocation}
\begin{aligned}
\underset{\alpha_1,\alpha_2,P_1,P_2}{\text{minimize}}  \quad & \sum_{i=1}^{2}  \alpha_i \frac{2P_i}{\kappa_i} + \mu_i  \\
\text{subject to} \, \quad & \sum_{i=1}^{2} \alpha_i B_{i} \log_2 \left( 1 + c_i P_i^2  \right) = R ,\\
& 0 \leq P_i, \,0 \le \alpha_i \le 1. 
\end{aligned}
\end{equation}

Solving this problem is equivalent to maximizing the EE of the dual-band BS, defined as
\begin{equation} \label{eq:EE-dualband}
\mathrm{EE}_{\text{dual}}  = \frac{\alpha_1 B_{1} \log_2 \left( 1 + c_1 P_1^2 \right) + \alpha_2 B_{2} \log_2 \left( 1 + c_2 P_2^2 \right)}{\alpha_1 \frac{2 P_1}{\kappa_1}  + \alpha_2  \frac{2 P_2}{\kappa_2} +\mu_1 + \mu_2},
\end{equation}
under the given rate constraint.
A first major step towards solving the problem is the following main result.

\begin{theorem} \label{theorem1}
Suppose the maximum EE in Band~1 is higher than in Band~2: $\mathrm{EE}_{1,\mathrm{opt}}> \mathrm{EE}_{2,\mathrm{opt}}$. To maximize the dual-band EE in \eqref{eq:EE-dualband}, we should activate the secondary band (i.e.,  $\alpha_2 > 0$) if and only if the data rate requirement $R$ exceeds the threshold $R_{\mathrm{bp}}$ given by
\begin{equation} \label{eq:Rbp}
R_{\mathrm{bp}}=B_1\log_2(1+c_1P_{1,\mathrm{bp}}^2).
\end{equation}
The transmit power $P_{1,\mathrm{bp}}$ in Band~1 at this threshold is 
\begin{equation} \label{eq:P1bp}
P_{1,\mathrm{bp}}=\left(\frac{\kappa_1B_1\sqrt{u^*-1}}{\kappa_2B_2\sqrt{c_2}\ln(u^*)}\right)+\sqrt{\left(\frac{\kappa_1B_1\sqrt{u^*-1}}{\kappa_2B_2\sqrt{c_2}\ln(u^*)}\right)^2-\frac{1}{c_1}}
\end{equation}
and $\alpha_1=1$.
\end{theorem}

\begin{IEEEproof}
Since $\mathrm{EE}_{1,\mathrm{opt}}> \mathrm{EE}_{2,\mathrm{opt}}$, the PC is minimized by only using Band~1 for $R \leq R_{1}^*$ and selecting $\alpha_1 = R/R_{1}^*$. Hence, the threshold $R_{\mathrm{bp}}$ must be larger than $R_{1}^*$, in which case the primary band reaches continuous operation ($\alpha_1=1$) and we must increase the transmit power $P_1$ beyond the optimal duty-cycled power level $P_1^*$.

To deliver a specific data rate $R_1>R_{1}^*$ in Band~1 with minimal PC, we need to solve the equation
$R_1 = B_{1} \log_2 \left( 1 + c_1 P_1^2  \right)$ for $P_1$, resulting in
\begin{equation} \label{eq:region2}
    P_1(R_1) = \sqrt{c_1^{-1}(2^{R_1 / B_{1}} - 1)}.
\end{equation}
The PC of Band~1 is $P_{\mathrm{avg},1}(R_1) = 2 P_1(R_1) / \kappa_1+\mu_1$. Taking the derivative of this PC expression with respect to $R_1$ yields
\begin{equation}
\frac{dP_{\mathrm{avg},1}}{dR_1} = \frac{2}{\kappa_1} \frac{d}{dR_1} \sqrt{\frac{2^{R_1 / B_{1}} - 1}{c_1}} = \frac{(1 + c_1 P_1^2) \ln (2)}{\kappa_1 B_{1} c_1 P_1},
\end{equation}
and it indicates how quickly the PC grows if we increase $R_1$.

It is desirable to turn on the second band when the marginal increase in PC is smaller in Band~2 than in Band~1.
The operation in Band~2 will begin using the EE-optimal operating point.
Recall that $P_2^* = \sqrt{(u^* - 1)/c_2}$, as defined in \eqref{eq:Pstar}. 
Because Band~2 is operated by increasing the activity factor from $0$ towards $1$, the  average PC and the provided data rate in Band~2 scale linearly with the activity factor $\alpha_2$:
\begin{align}
    P_{\mathrm{avg},2}(\alpha_2) &= \alpha_2 \frac{2P_2^*}{\kappa_2} + \mu_2, \label{eq:P_avg_alpha} \\
    \alpha_2(R_2) &= \frac{R_2}{B_2 \log_2(1+c_2(P_2^*)^2)}  , \label{eq:R_alpha}
\end{align}
where $R_2$ denotes the rate in Band~2.

Therefore, using the chain rule to calculate, the derivative of $P_{\mathrm{avg},2}(\alpha_2)$ with respect to $R_2$ yields
\begin{align} \nonumber
    \frac{dP_{\mathrm{avg},2}}{dR_2} &= \frac{dP_{\mathrm{avg},2}}{d\alpha_2} \cdot \frac{d\alpha_2}{dR_2} = \frac{2P_2^*/\kappa_2}{B_2\log_2(1+c_2(P_2^*)^2)} \\
    &=\frac{2\sqrt{u^*-1}}{\kappa_2B_{2}\sqrt{c_2}\log_2(u^*)},
\end{align}
where the last expression follows by replacing $P_2^*$ with the expression in \eqref{eq:Pstar}.

Equating the derivatives ($\frac{dP_{\mathrm{avg},1}}{dR_1} = \frac{dP_{\mathrm{avg},2}}{dR_2}$) for the two bands gives
\begin{equation}
\frac{(1+c_1P_{1,\mathrm{bp}}^2)\ln (2)}{\kappa_1B_{1}c_1P_{1,\mathrm{bp}}}=\frac{2\sqrt{u^*-1}}{\kappa_2B_{2}\sqrt{c_2}\log_2(u^*)}.
\end{equation}
Using that $\ln(2)\log_2(u^*) = \ln(u^*)$, we rearrange the terms into a standard quadratic equation 
\begin{equation}
c_1P_{1,\mathrm{bp}}^2-2\left(\frac{\kappa_1B_{1}c_1\sqrt{u^*-1}}{\kappa_2B_{2}\sqrt{c_2}\ln(u^*)}\right)P_{1,\mathrm{bp}}+1=0.
\end{equation}
Solving it directly yields the closed-form expression for $P_{1,\mathrm{bp}}$ in \eqref{eq:P1bp}. Substituting this power value back into the rate expression of Band~1 yields the rate threshold $R_\mathrm{bp}$, concluding the proof.
\end{IEEEproof}

A consequence is that for rates $R_1^* \leq R \leq R_{\mathrm{bp}}$, the BS operates in \textbf{Region 2}. In this region, Band~1 serves all traffic, and the power allocation in \eqref{eq:region2} shows that the PC scales exponentially with the rate requirement.
Theorem~\ref{theorem1} proves that once the target data rate $R$ exceeds the activation threshold $R_\mathrm{bp}$ in \eqref{eq:Rbp}, the system enters a dual-band operating regime. Band~1 (sub-6 GHz) operates continuously ($\alpha_1 = 1$) at the fixed transmit power $ P_{1,\mathrm{bp}}$ in \eqref{eq:P1bp}, providing a constant data rate contribution of exactly $R_{\mathrm{bp}}$.
The residual data rate, defined as $R_2 = R - R_{\mathrm{bp}}$, is offloaded to Band~2 (sub-THz), and managed by gradually increasing the activity factor $\alpha_2$.
The transmit power is fixed at $P_2^*$ given by \eqref{eq:Pstar}, which gives the highest EE during data transmission, and the rate provided by Band~2 while actively transmitting is $B_2 \log_2(u^*)$ as established in \eqref{eq:optimal_rate}. 
The activity factor $\alpha_2$, required to meet this rate in Band~2, is 
\begin{equation}
    \alpha_2(R) = \frac{R_2}{B_2 \log_2(u^*)}.
\end{equation}

The total average PC of the BS in this dual-band operation is the sum of the fully active Band~1 power and the duty-cycled Band~2 power:
\begin{equation}
    P_{\mathrm{avg}}(R) = \left( \frac{2 P_{1,\mathrm{bp}}}{\kappa_1} + \mu_1 \right) + \left( \alpha_2(R) \frac{2 P_2^*}{\kappa_2} + \mu_2 \right).
\end{equation}

\begin{table}[t!]
  \begin{center}
    \caption{Simulation Parameters}
    \label{tab:table1}
    \begin{tabular}{|l|r|} 
     \hline
      \textbf{Parameter} & \textbf{Value} (Band~1, Band~2) \\
      \hline
      Passive circuit power $\mu_1, \mu_2$ & $1000$\,mW, $1000$\,mW \\
      Transceiver chain power: $D_{0,1} \, D_{0,2}$ & $100$\,mW, $10$\,mW  \\
      Sample processing energy: $\nu_1,\nu_2$ & $10^{-10}, 10^{-11}$\,J/sample \\
      Power amplifier efficiency: $\kappa_1, \kappa_2$ & 0.4, 0.05 \\
      Carrier bandwidths $B_1, B_2$ & $100$\,MHz, $400$\,MHz \\
      Receiver noise power spectral density: $N_0$ & $-174$\,dBm/Hz \\
      Channel gain (Band~1, Band~2): $\beta_1, \beta_2$ & $-126$\,dB, $-153$\,dB \\
      Receive antennas  $N_1, N_2$ & $4$, $64$ \\
       \hline
    \end{tabular}
  \end{center}
\end{table}

Theorem~\ref{theorem1} identifies the rate $R_{\mathrm{bp}}$ where Band~2 is turned on and how the residual rate $R-R_{\mathrm{bp}}$ is handled by gradually increasing $\alpha_2 \in (0,1]$.
However, we reach a saturation data rate, $R_{\mathrm{sat}}$, when Band~2 reaches continuous operation ($\alpha_2 = 1$). This rate is
\begin{equation}
    R_{\mathrm{sat}} = R_{\mathrm{bp}} + B_{2} \log_2(u^*).
\end{equation}
Hence, we identify \textbf{Region 3} where $R_{\mathrm{bp}}< R \leq R_{\mathrm{sat}}$. In this region, the allocated transmit power on both carriers remains constant ($P_{1,\mathrm{bp}}$ and $P_2^*$). To illustrate these analytical findings, in Fig.~\ref{fig:dualbandPC} we show how the total PC increases with the required data rate using the simulation parameters in Table~\ref{tab:table1}. The curve demonstrates a linear scaling in Region 1, an exponential scaling in Region 2, and back to linear scaling in Region 3 as the sub-THz band absorbs the residual traffic. The corresponding optimal activity factors ($\alpha_1$ and $\alpha_2$) for the required data rate $R$ are shown in Fig.~\ref{fig:activityfactors}.

\begin{figure}[t!]
    \centering
    \subfloat[Power Consumption for different rate requirements.]{
        \includegraphics[width=\columnwidth]{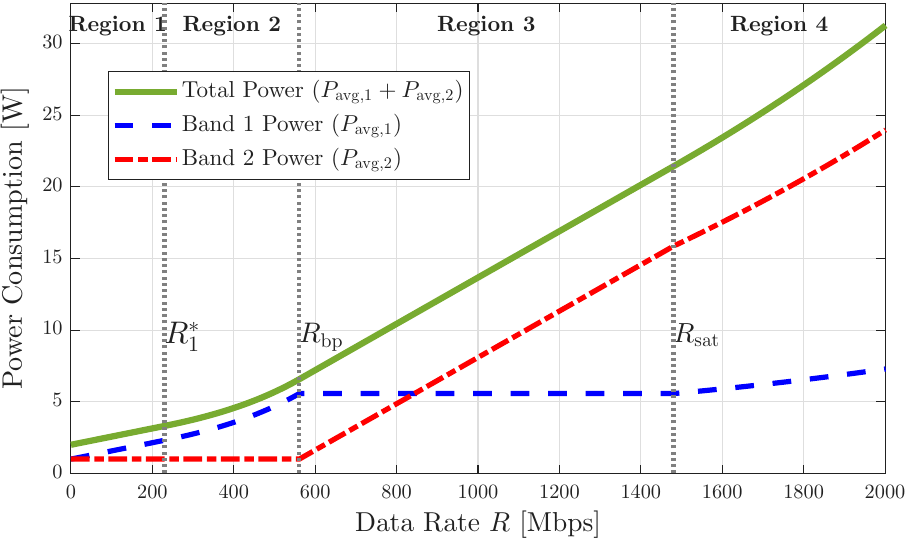}
        \label{fig:dualbandPC}
    }
    \\ \vspace{-2mm} 
    \subfloat[Optimal activity factors $\alpha_1$ and $\alpha_2$ for different rate requirements.]{
        \includegraphics[width=\columnwidth]{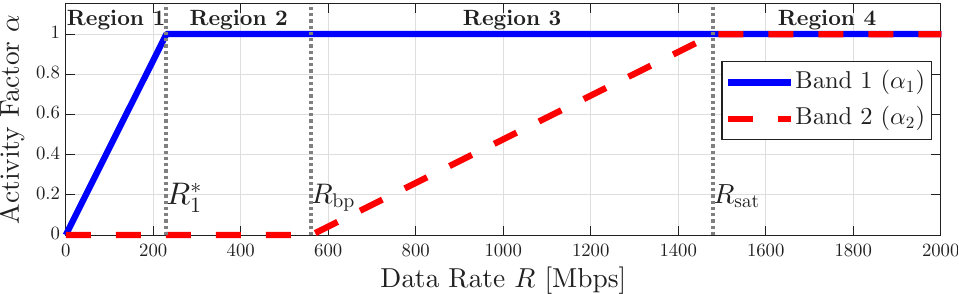}
        \label{fig:activityfactors}
    }
    \caption{Dual-band base station power consumption (a) and activity factor (b) configuration across varying data rate requirements.}
    \label{fig:combined_performance}
    \vspace{-4mm}
\end{figure}

\begin{table*}[t]
\centering
\caption{Dual-Band Operation Regions for Band~1 and Band~2 Carriers.}
\label{tab:dualband_regions}
\begin{tabular}{|c|c|c|c|c|}
\hline
\textbf{Region} & \textbf{Rate Range $R$} & \textbf{Band~1 Status} & \textbf{Band~2 Status} & \textbf{Total PC Scaling} \\
\hline
1 & $0 < R \leq R_{1}^*$ & Partially Active ($\alpha_1 < 1$) & Off ($\alpha_2=0$) & Linear (Band~1 duty-cycled) \\
\hline
2 & $R_1^* < R \leq R_{\mathrm{bp}}$ & Fully active ($\alpha_1=1$) & Off ($\alpha_2=0$) & Exponential (Band~1 saturated) \\
\hline
3 & $R_{\mathrm{bp}} < R \leq R_{\mathrm{sat}} $ & Fully active ($\alpha_1=1$) & Partially Active $(\alpha_2<1)$ & Linear (Band~2 duty-cycled) \\
\hline
4 & $R > R_{\mathrm{sat}}$ & Fully active ($\alpha_1=1$) & Fully active ($\alpha_2=1$) & Exponential (Both saturated) \\
\hline
\end{tabular}
\end{table*}

\subsection{Optimal Traffic Allocation in Very High-Rate Scenarios}
When the required data rate exceeds the saturation threshold $R_{\mathrm{sat}}$, the system enters \textbf{Region 4}. As shown in Fig.~\ref{fig:dualbandPC}, both bands must remain fully active continuously (i.e., $\alpha_1 = 1$ and $\alpha_2 = 1$), and we must allocate power that deviates from the most energy-efficient operating point in both bands. In this region, the general optimization problem in \eqref{eq:quad_power_allocation} simplifies to
\begin{equation}
\begin{aligned}
    \underset{P_1,P_2}{\text{minimize}} \quad & \sum_{i=1}^{2}  \frac{2P_i}{\kappa_i} + \mu_i \\
    \text{subject to} \quad & B_{1} \log_2 \left( 1 + c_1 P_1^2 \right) + B_{2} \log_2 \left( 1 + c_2 P_2^2 \right) = R, \\
    & P_1 \ge P_{1,\mathrm{bp}}, \quad P_2 \ge P_2^*.
\end{aligned}
\end{equation}

To determine the optimal power allocation $(P_1, P_2)$ across the two bands, we formulate the Lagrangian function:
\begin{align} \nonumber
    &\mathcal{L}(P_1, P_2, \lambda) = \frac{2 P_1}{\kappa_1} + \frac{2 P_2}{\kappa_2} +\mu_1+\mu_2\\ &- \lambda \Big( B_{1} \log_2(1 + c_1 P_1^2) + B_{2} \log_2(1 + c_2 P_2^2) - R \Big),
\end{align}
where $\lambda$ is the Lagrange multiplier associated with the sum-rate constraint. The first-order Karush-Kuhn-Tucker (KKT) optimality conditions are obtained by taking the partial derivatives of $\mathcal{L}$ with respect to $P_1$ and $P_2$, and equating them to zero:
\begin{align}
    \frac{\partial \mathcal{L}}{\partial P_1} &= \frac{2}{\kappa_1} - \lambda \left( \frac{B_{1}}{\ln (2)} \frac{2 c_1 P_1}{1 + c_1 P_1^2} \right) = 0, \\
    \frac{\partial \mathcal{L}}{\partial P_2} &= \frac{2}{\kappa_2} - \lambda \left( \frac{B_{2}}{\ln (2)} \frac{2 c_2 P_2}{1 + c_2 P_2^2} \right) = 0.
\end{align}
By isolating the Lagrange multiplier $\lambda$ in both KKT equations, we reveal that the optimal power allocation mandates
\begin{equation} \label{eq:lambdaequation}
    \lambda = \frac{(1 + c_1 P_1^2) \ln (2)}{\kappa_1 B_{1} c_1 P_1} = \frac{(1 + c_2 P_2^2) \ln (2)}{\kappa_2 B_{2} c_2 P_2}.
\end{equation}

As a consequence, the optimal power allocation $(P_1,P_2)$ for any $R > R_{\mathrm{sat}}$ is defined by the solution to the following system of nonlinear equations:
\begin{equation}
    \begin{cases} 
\frac{1 + c_1 P_1^2}{\kappa_1 B_{1} c_1 P_1} = \frac{1 + c_2 P_2^2}{\kappa_2 B_{2} c_2 P_2}\label{eq:nonlin1} \\
B_{1} \log_2(1 + c_1 P_1^2) + B_{2} \log_2(1 + c_2 P_2^2) = R.
\end{cases}
\end{equation}

A closed-form analytical expression for $P_1(R)$ and $P_2(R)$ cannot be found. However, the optimal power allocation in this regime can be evaluated using standard numerical root-finding algorithms, such as a line search over $\lambda$, followed by solving \eqref{eq:lambdaequation} as a quadratic equation.
Nevertheless, the expressions in \eqref{eq:nonlin1} can be used to study the asymptotic power allocation for extremely high data rates.

\begin{theorem}
In the fully active dual-band regime ($\alpha_1 = 1, \alpha_2 = 1$), as the target data rate $R \to \infty$, the optimal power allocation between the two bands becomes independent of the channel conditions and has the following asymptotically optimal ratio:
\begin{equation}
    \lim_{R \to \infty} \frac{P_1(R)}{P_2(R)} = \frac{\kappa_1 B_1}{\kappa_2 B_2}.
    \end{equation}
\end{theorem}

\begin{IEEEproof}
    We can rewrite the first equation in \eqref{eq:nonlin1} as
\begin{align} \nonumber
     &\frac{1}{\kappa_1 B_{1} c_1 P_1} + \frac{P_1}{\kappa_1 B_{1}} = \frac{1}{\kappa_2 B_{2} c_2 P_2} + \frac{P_2}{\kappa_2 B_{2}} \\
     &\Rightarrow \frac{P_1}{P_2} = \frac{\kappa_1 B_1}{\kappa_2 B_2}+\frac{\kappa_1 B_{1}}{\kappa_2 B_{2} c_2 P_2^2}  - \frac{\kappa_1 B_{1}}{\kappa_1 B_{1} c_1 P_1P_2} . \label{eq:ratio-derivation}
\end{align}
As the rate constraint $R \to \infty$, the required transmission powers $P_1, P_2 \to \infty$. Consequently, the last two terms in \eqref{eq:ratio-derivation} go to zero in the asymptotic regime, and the first term remains as the asymptotic limit.
\end{IEEEproof}

This theorem demonstrates that, in the extremely high rate regime, the optimal transmit power distribution across the two bands is determined exclusively by the PA efficiencies and bandwidths: $P_i = \kappa_i B_i \rho$ for some value of $\rho$. An overview of the optimal activity factors, and power consumption scaling across all four operational regions is provided in Table~\ref{tab:dualband_regions}.

\section{Secondary Band Activation Threshold}
In the preceding analysis, the dual-band PC model incurs the passive circuit power ($\mu_1$ and $\mu_2$) of both bands. To precisely identify when the proposed dual-band architecture becomes more energy-efficient than a single-band baseline, we must determine the crossover rate $R_\mathrm{cross}$. By equating the PC of the two architectures, we establish this formally in the following theorem.

\begin{theorem}
\label{thm:activation_threshold}
The data rate threshold $R_{\mathrm{cross}}$ at which the proposed dual-band architecture strictly outperforms a single-band baseline in terms of EE is obtained by numerically solving the following equation for $R$:
\begin{equation}
\frac{2}{\kappa_1 \sqrt{c_1}}\sqrt{2^{R/B_1}-1} = \frac{2 P_{1,\mathrm{bp}}}{\kappa_1} + \frac{R - R_{\mathrm{bp}}}{R_2^*} \left( \frac{2 P_2^*}{\kappa_2} \right) + \mu_2 \label{eq:exact_cross}.
\end{equation}
This equation admits an approximate closed-form solution
\begin{equation}
R_{\mathrm{cross}} \approx -C \left(\frac{\kappa_2 R_2^*}{2 P_2^*}\right) - \frac{2 B_1}{\ln 2} W_{-1}(Z), \label{eq:approx_cross}
\end{equation}
where $W_{-1}(\cdot)$ is the lower branch of the Lambert W function, $C = \frac{2 P_{1,\mathrm{bp}}}{\kappa_1} + \mu_2 - \left(\frac{2 P_2^*}{\kappa_2 R_2^*}\right) R_{\mathrm{bp}}$, and the argument $Z = - \frac{\kappa_2 R_2^* \ln 2}{2 P_2^* B_1 \kappa_1 \sqrt{c_1}} \exp\left( - \frac{C \kappa_2 R_2^* \ln 2}{4 P_2^* B_1} \right)$.
\end{theorem}

\begin{IEEEproof}
Equating the PC of the Band~1 only baseline to that of the proposed dual-band system operating in Region 3 (where Band~1 is fully active at $P_{1,\mathrm{bp}}$ and Band~2 is duty-cycling):
\begin{equation}
\frac{2}{\kappa_1}\sqrt{\frac{2^{R/B_1}-1}{c_1}} + \mu_1 = \frac{2 P_{1,\mathrm{bp}}}{\kappa_1} + \mu_1 + \mu_2 + \frac{R - R_{\mathrm{bp}}}{R_2^*} \left( \frac{2 P_2^*}{\kappa_2} \right) .
\end{equation}
Canceling the $\mu_1$ term yields \eqref{eq:exact_cross}. Because the data rate is exceptionally high in Region 3, we apply the approximation $2^{R/B_1} - 1 \approx 2^{R/B_1}$. This allows the left side of \eqref{eq:exact_cross} to be rewritten using the natural exponential, giving:
\begin{equation}
\frac{2}{\kappa_1 \sqrt{c_1}} \exp\left(\frac{R \ln 2}{2 B_1}\right) \approx  \frac{2 P_{1,\mathrm{bp}}}{\kappa_1} + \mu_2 + \frac{R - R_{\mathrm{bp}} }{ R_2^*}\left(\frac{2 P_2^*}{\kappa_2}\right) .\label{eq:approx_exp}
\end{equation}

To analytically isolate $R$, we substitute the constant $C$ into \eqref{eq:approx_exp} and manipulate the equation into the canonical form ($X = Y e^Y$), giving:
\begin{align} \label{eq:lambert_form}
&- \left( \frac{\kappa_2 R_2^* \ln 2}{2 P_2^* B_1 \kappa_1 \sqrt{c_1}} \right) \exp\left( - \frac{C \kappa_2 R_2^* \ln 2}{4 P_2^* B_1} \right) =  \\
-&\frac{\ln 2}{2 B_1}\left( R + C \frac{\kappa_2 R_2^*}{2 P_2^*} \right) \exp\left( -\frac{\ln 2}{2 B_1}\left( R + C \frac{\kappa_2 R_2^*}{2 P_2^*} \right) \right). \nonumber
\end{align}
Applying the secondary branch of the Lambert W function ($W_{-1}$) to the negative resulting argument yields:
\begin{align}
W_{-1}&\left( - \frac{\kappa_2 R_2^* \ln 2}{2 P_2^* B_1 \kappa_1 \sqrt{c_1}} \exp\left( - \frac{C \kappa_2 R_2^* \ln 2}{4 P_2^* B_1} \right) \right) = \\
&-\frac{\ln 2}{2 B_1}\left( R + C \frac{\kappa_2 R_2^*}{2 P_2^*} \right). \nonumber
\end{align}

Finally, algebraically solving for $R$ yields the approximate closed-form threshold $R_\mathrm{cross}$ presented in \eqref{eq:approx_cross}.
\end{IEEEproof}

This theorem determines the exact rate where a dual-band BS outperforms a single-band BS from an EE viewpoint. Fig.~\ref{fig:dualbandEE} compares the proposed system's EE against a Band~1-only baseline and a heuristic bandwidth-proportional split. At low rates, passive circuit powers ($\mu_1, \mu_2$) dominate the EE. For rates lower than $R_\mathrm{cross}$, the baseline ``Band~1 Only'' system is the most energy efficient. But as the required rate increases, it suffers from capacity saturation, causing its EE to degrade rapidly. The proposed dual-band architecture avoids this by reallocating traffic to Band~2, making the dual-band setup more efficient in high-rate scenarios. Fig.~\ref{fig:dualbandEE} also evaluates a $10$~dB loss in $\beta_2$, representing physical blockages or imperfect CSI. This attenuation reduces Band~2's EE, shifting $R_{\mathrm{bp}}$ and $R_{\mathrm{cross}}$ rightward, forcing the BS to rely on Band~1 longer. The validity of Theorem~\ref{thm:activation_threshold} is also confirmed by Fig.~\ref{fig:dualbandEE}. The exact numerical crossing point $R_\mathrm{cross}$ coincides perfectly with the intersection of the Band~1 only and proposed dual-band curves. Furthermore, the marker for the approximate closed-form threshold \eqref{eq:approx_cross} aligns tightly with the exact solution.

\begin{figure}[t!]
	\centering \vspace{-2mm}
	\includegraphics[width=\columnwidth]{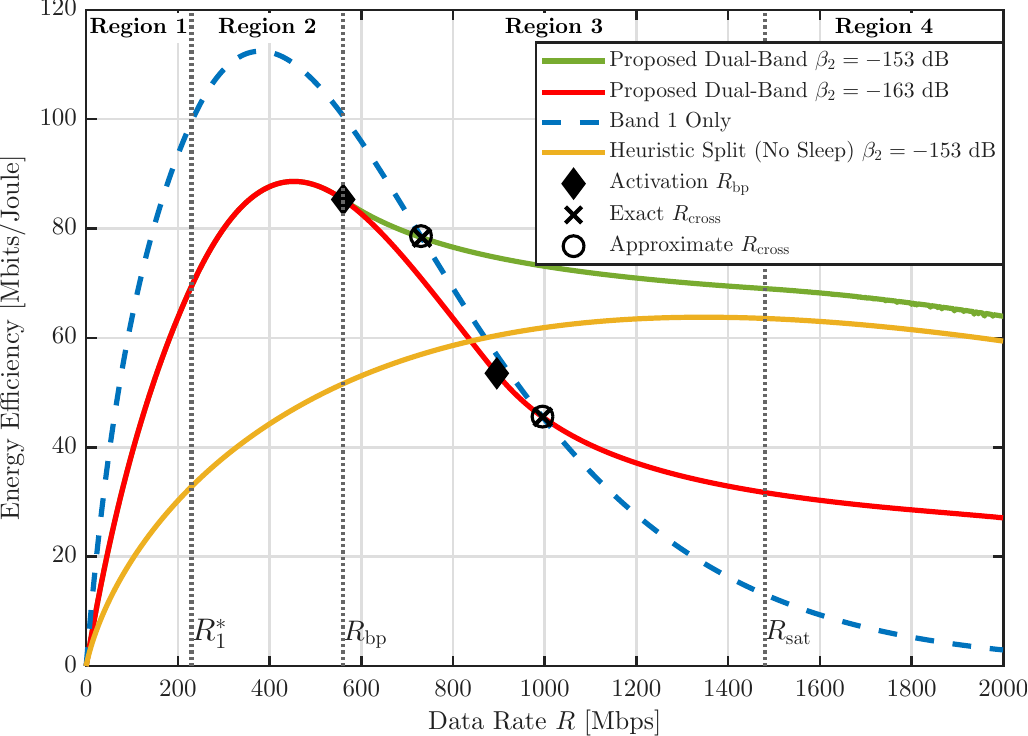}
	\caption{The EE of the proposed dual band BS for different rate requirements compared to a heuristic split and a single low band.}
	\label{fig:dualbandEE} \vspace{-2mm}
\end{figure}

\section{Conclusion}

In this paper, we investigated the fundamental EE limits of a dual-band BS combining a coverage-oriented sub-6 GHz carrier with a capacity-oriented sub-THz carrier. We demonstrated analytically that optimal EE across the entire traffic spectrum requires the hardware configuration (power and antennas) and sleep mode (duty cycle) to be jointly optimized. Our analysis identified four distinct operational regimes and yielded closed-form analytical expressions that dictate precisely when the sub-THz band should be activated, and how to allocate traffic across the bands in the dual-band case. Ultimately, by linking the load directly to the activity factor this framework provides mathematical guidelines for sleep-mode management and optimal operating points in future green multi-band systems.

\bibliographystyle{IEEEtran}

\bibliography{IEEEabrv,Referenser}

\begin{thebibliography}{10}
\providecommand{\url}[1]{#1}
\csname url@samestyle\endcsname
\providecommand{\newblock}{\relax}
\providecommand{\bibinfo}[2]{#2}
\providecommand{\BIBentrySTDinterwordspacing}{\spaceskip=0pt\relax}
\providecommand{\BIBentryALTinterwordstretchfactor}{4}
\providecommand{\BIBentryALTinterwordspacing}{\spaceskip=\fontdimen2\font plus
\BIBentryALTinterwordstretchfactor\fontdimen3\font minus
  \fontdimen4\font\relax}
\providecommand{\BIBforeignlanguage}[2]{{%
\expandafter\ifx\csname l@#1\endcsname\relax
\typeout{** WARNING: IEEEtran.bst: No hyphenation pattern has been}%
\typeout{** loaded for the language `#1'. Using the pattern for}%
\typeout{** the default language instead.}%
\else
\language=\csname l@#1\endcsname
\fi
#2}}
\providecommand{\BIBdecl}{\relax}
\BIBdecl

\bibitem{kolta2024measuring}
E.~Kolta and T.~Hatt, ``{Measuring Energy Efficiency of Mobile Networks
  2024},'' GSMA Intelligence, Tech. Rep., feb 2024.

\bibitem{ericssonmobilityreport2023}
\BIBentryALTinterwordspacing
Ericsson, ``{Ericsson Mobility Report June 2023},'' 2023. [Online]. Available:
  \url{https://www.ericsson.com/en/reports-and-papers/mobility-report/reports/}
\BIBentrySTDinterwordspacing

\bibitem{bjornson2017massive}
E.~Bj{\"o}rnson, J.~Hoydis, and L.~Sanguinetti, ``Massive {MIMO} networks:
  Spectral, energy, and hardware efficiency,'' \emph{Foundations and
  Trends{\textregistered} in Signal Processing}, vol.~11, no. 3-4, pp.
  154--655, 2017.

\bibitem{Rappaport2019a}
T.~S. Rappaport \emph{et~al.}, ``Wireless communications and applications above
  100 {GHz}: Opportunities and challenges for {6G} and beyond,'' \emph{IEEE
  Access}, vol.~7, pp. 78\,729--78\,757, 2019.

\bibitem{RIS-EE}
C.~Huang, A.~Zappone, G.~C. Alexandropoulos, M.~Debbah, and C.~Yuen,
  ``Reconfigurable intelligent surfaces for energy efficiency in wireless
  communication,'' \emph{IEEE Transactions on Wireless Communications},
  vol.~18, no.~8, pp. 4157--4170, 2019.

\bibitem{lopez2022survey}
D.~L{\'o}pez-P{\'e}rez, A.~De~Domenico, N.~Piovesan, G.~Xinli, H.~Bao,
  S.~Qitao, and M.~Debbah, ``A survey on {5G} radio access network energy
  efficiency: Massive {MIMO}, lean carrier design, sleep modes, and machine
  learning,'' \emph{IEEE Communications Surveys \& Tutorials}, vol.~24, no.~1,
  pp. 653--697, 2022.

\bibitem{Tombaz2014a}
S.~Tombaz, S.-w. Han, K.~W. Sung, and J.~Zander, ``Energy efficient network
  deployment with cell {DTX},'' \emph{IEEE Communications Letters}, vol.~18,
  no.~6, pp. 977--980, 2014.

\bibitem{rottenberg2024information}
F.~Rottenberg, ``Information-theoretic study of time-domain energy-saving
  techniques in radio access,'' \emph{IEEE Transactions on Green Communications
  and Networking}, vol.~9, no.~2, pp. 605--620, 2024.

\bibitem{peschiera2025optimizing}
E.~Peschiera, Y.~Agram, F.~Quitin, L.~Van~der Perre, and F.~Rottenberg, ``On
  optimizing time-, space-and power-domain energy-saving techniques for sub-6
  ghz massive mimo base stations,'' in \emph{2025 IEEE 26th International
  Workshop on Signal Processing and Artificial Intelligence for Wireless
  Communications (SPAWC)}.\hskip 1em plus 0.5em minus 0.4em\relax IEEE, 2025,
  pp. 1--5.

\bibitem{tataria20216g}
H.~Tataria, M.~Shafi, A.~F. Molisch, M.~Dohler, H.~Sj{\"o}land, and
  F.~Tufvesson, ``6g wireless systems: Vision, requirements, challenges,
  insights, and opportunities,'' \emph{Proceedings of the IEEE}, vol. 109,
  no.~7, pp. 1166--1199, 2021.

\bibitem{lopez2021energy}
D.~L{\'o}pez-P{\'e}rez, A.~De~Domenico, N.~Piovesan, X.~Geng, H.~Bao, and
  M.~Debbah, ``Energy efficiency of multi-carrier massive {MIMO} networks:
  Massive {MIMO} meets carrier aggregation,'' in \emph{2021 IEEE Global
  Communications Conference (GLOBECOM)}.\hskip 1em plus 0.5em minus 0.4em\relax
  IEEE, 2021, pp. 01--07.

\bibitem{enqvist2024fundamentals}
A.~Enqvist, {\"O}.~T. Demir, C.~Cavdar, and E.~Bj{\"o}rnson, ``Fundamentals of
  energy-efficient wireless links: Optimal ratios and scaling behaviors,'' in
  \emph{2024 IEEE 99th Vehicular Technology Conference (VTC2024-Spring)}.\hskip
  1em plus 0.5em minus 0.4em\relax IEEE, 2024, pp. 1--6.

\bibitem{peesapati2021analytical}
S.~K.~G. Peesapati, M.~Olsson, M.~Masoudi, S.~Andersson, and C.~Cavdar, ``An
  analytical energy performance evaluation methodology for 5g base stations,''
  in \emph{2021 17th International conference on wireless and mobile computing,
  networking and communications (WiMob)}.\hskip 1em plus 0.5em minus
  0.4em\relax IEEE, 2021, pp. 169--174.

\bibitem{Enqvist_Fundamentals_TBS}
A.~Enqvist, {\"O}.~T. Demir, C.~Cavdar, and E.~Bj{\"o}rnson, ``{Fundamentals of
  Energy-Efficient Hardware Configurations for Wireless Links with Sleep
  Modes},'' 2026, to be submitted, unpublished manuscript.

\end{thebibliography}

\end{document}